\def\random{0}
\documentclass[11pt,letterpaper]{article}

\usepackage[margin=1in,letterpaper]{geometry}
\usepackage[colorlinks,urlcolor=blue,citecolor=blue,linkcolor=blue]{hyperref}
\usepackage{fullpage}
\usepackage{subcaption}
\usepackage{graphicx}
\usepackage{epstopdf}
\graphicspath{{./figs/}}

\usepackage{multirow}
\usepackage{comment}
\usepackage{amsfonts,amsmath,mathtools,amssymb,amsthm}

\usepackage{xcolor}
\usepackage[ruled,vlined]{algorithm2e}

\usepackage{bbm,bm}
\usepackage{xspace,prettyref}
\usepackage{color}
\usepackage{dsfont}

\title{Testing Hamiltonicity (and other problems) in Minor-Free Graphs}

\author{
Reut Levi\thanks{Efi Arazi School of Computer Science,  The Interdisciplinary Center, Israel. Email: {\tt reut.levi1@idc.ac.il}.}
\and
Nadav Shoshan\thanks{Efi Arazi School of Computer Science,  The Interdisciplinary Center, Israel.  Email: {\tt nadav.shoshan1@post.idc.ac.il}.}
}

\newcommand{\ham}{\rm HAM}
\newcommand{\eps}{\epsilon}
\newcommand{\eqdef}{\stackrel{\rm def}{=}}

\newcommand{\BE}{\begin{enumerate}} \newcommand{\EE}{\end{enumerate}}
\newcommand{\BI}{\begin{itemize}} \newcommand{\EI}{\end{itemize}}
\newcommand{\BDes}{\begin{description}}\newcommand{\EDes}{\end{description}}
\newtheorem{alg}{Algorithm}
\newcommand{\BA}{\begin{alg}} \newcommand{\EA}{\end{alg}}
\newtheorem{thm}{Theorem}
\newcommand{\BT}{\begin{thm}} \newcommand{\ET}{\end{thm}}

\newtheorem{lem}{Lemma}      % A counter for Lemmas
\newcommand{\BL}{\begin{lem}} \newcommand{\EL}{\end{lem}}

\newtheorem{clm}[lem]{Claim}
\newcommand{\BCM}{\begin{clm}} \newcommand{\ECM}{\end{clm}}

\newtheorem{fct}[lem]{Fact}
\newcommand{\BF}{\begin{fct}} \newcommand{\EF}{\end{fct}}

\newtheorem{techcor}[thm]{Corollary}
\newcommand{\BCo}{\begin{techcor}} \newcommand{\ECo}{\end{techcor}}

\newtheorem{cor}[thm]{Corollary}      % counter AS FOR Theorems
\newcommand{\BC}{\begin{cor}} \newcommand{\EC}{\end{cor}}
\newtheorem{prop}[thm]{Proposition}     % A counter AS FOR Thms
\newcommand{\BP}{\begin{prop}} \newcommand {\EP}{\end{prop}}
\newtheorem{conj} {Conjecture}      % counter AS FOR Theorems
\newcommand{\BCJ}{\begin{conj}} \newcommand{\ECJ}{\end{conj}}
\newtheorem{defn}{Definition}         % A counter for Definition
\newcommand{\BD}{\begin{defn}} \newcommand{\ED}{\end{defn}}

\newcommand{\bx}{{\bold x}}
\newcommand{\dist}{\delta_{\rm HAM}}
\newcommand{\by}{{\bold y}}
\newcommand{\clip}[2]{{\rm cl}(#1, #2)}

\newtheorem{theorem}{Theorem}

\newtheorem{lemma}{Lemma}
\newtheorem{claim}[lemma]{Claim}
\newtheorem{definition}[lemma]{Definition}

\newtheorem{remark}{Remark}

\newcommand{\calA}{{\mathcal A}}
\newcommand{\calC}{{\mathcal C}}
\newcommand{\calF}{{\mathcal F}}

\newcommand{\calP}{{\mathcal P}}
\newcommand{\wvec}[2]{\mathbf{p}_{#1, #2}}

\newcommand{\poly}{{\rm poly}}
\newcommand{\E}{{\rm E}}
\newcommand{\w}{w}
\newcommand{\MST}{{\rm MST}}
\newcommand{\MSF}{{\rm MSF}}

\newcommand{\W}{W}  % Set of centers (was A)
\renewcommand{\Pr}{\mathrm{Pr}}

\newcommand{\Var}{\mathrm{Var}}

\newcommand{\s}[1]{\left\lvert #1 \right\rvert}

\renewcommand{\d}{d} % Notation for distance. Was $d(\cdot,\cdot)$, but $d$ is degree

\newcommand{\calT}{{\cal T}}
\newcommand{\sz}{k}  % Size parameter for decomposition (was $\ell$)
\newcommand{\heavy}{{\cal H}}
\newcommand{\light}{{\cal L}}
\newcommand{\wtd}{\Delta}
\newcommand{\wtG}{{\widetilde{G}}}

\newcommand{\mnote}[1]{{\color{red}$\spadesuit$} \marginpar{\tiny\bf
            \begin{minipage}[t]{0.5in}
              \raggedright #1
         \end{minipage}}}
         
\newcommand{\prw}[3]{p_{#1, #3}(#2)}
 \newcommand{\len}{\ell}

\begin{document}

\begin{titlepage}
\maketitle
\thispagestyle{empty}

\begin{abstract}
In this paper we provide sub-linear algorithms for several fundamental problems in the setting in which the input graph excludes a fixed minor, i.e., is a minor-free graph.
In particular, we provide the following algorithms for minor-free unbounded degree graphs. 
\begin{enumerate}
    \item A tester for Hamiltonicity with two-sided error with $\poly(1/\eps)$-query complexity, where $\eps$ is the proximity parameter.
    \item A local algorithm, as defined by Rubinfeld et al. (ICS 2011), for constructing a spanning subgraph with almost minimum weight, specifically, at most a factor $(1+\eps)$ of the optimum, with $\poly(1/\eps)$-query complexity.
\end{enumerate}
Both our algorithms use partition oracles, a tool introduced by Hassidim et al. (FOCS 2009), which are oracles that provide access to a partition of the graph such that the number of cut-edges is small and each part of the partition is small. 
The polynomial dependence in $1/\eps$ of our algorithms is achieved by combining the recent $\poly(d/\eps)$-query partition oracle of Kumar-Seshadhri-Stolman (ECCC 2021) for minor-free graphs with degree bounded by $d$.

For bounded degree minor-free graphs we introduce the notion of {\em covering partition oracles} which is a relaxed version of partition oracles 
and design a $\poly(d/\eps)$-time covering partition oracle for this family of graphs. 
Using our covering partition oracle we provide the same results as above (except that the tester for Hamiltonicity has one-sided error) for minor-free bounded degree graphs, 
as well as showing that any property which is monotone and additive (e.g. bipartiteness) can be tested in minor-free graphs by making $\poly(d/\eps)$-queries.

The benefit of using the covering partition oracle rather than the partition oracle in our algorithms is its simplicity and an improved polynomial dependence in $1/\eps$ in the obtained query complexity.
\end{abstract}

\end{titlepage}

\section{Introduction}
The family of minor-free graphs has been at the focus of attention ever since the theory of graph minors began many decades ago and has been drawing much attention in the field of computer science as well. 
Aside from being an important family that includes natural families of graphs such as planar graphs, it also has the appeal that some hard graph problems become easy when restricted to this family of graphs (e.g. Graph Isomorphism~\cite{HT72}).  

Minor-free graphs have been extensively studied also in the realm of sublinear algorithms and in particular property testing (see e.g.~\cite{BSS08,CGRSSS14,NS13,KY14,Ito16,BKN16,FLV18,KSS18,CS19,CMOS19,KSS19,KSS21}).
In particular, in the general-graph model~\cite{PR02} where there is much less body of work, compared to the bounded-degree graph model~\cite{GR02} and the dense graph model~\cite{GGR98}, these graphs draw attention as they allow for better characterization compared to general unbounded degree graphs.
A notable example is the result by Czumaj and Sohler~\cite{CS19} who recently gave a full characterization of properties that can be tested with one-sided error with query complexity that is independent of the size of the graph (i.e. {\em testable}).
They showed that the latter is possible if and only if testing the property can be reduced to testing for a finite family of finite forbidden subgraphs. 
This raises the question regarding the testability of properties that can not be reduced to testing for a finite family of finite forbidden subgraphs when we allow the tester to have two-sided error. A well known example of such property is the property of being Hamiltonian. 

Another question, which is also relevant for bounded degree graphs, is whether we can obtain algorithms with query complexity which is only polynomial in $1/\eps$ where $\eps$ is the proximity parameter. 
Newman and Sohler~\cite{NS13} showed that any property of bounded degree hyperfinite graphs and in particular minor-free graphs is testable. 
Their algorithm learns the graph up to modifications of $\eps d n$ edges with query complexity which is super-polynomial in $d$ and $\eps$. 
While this approach works for all properties of graphs, more efficient testers can be obtained for specific properties of graphs.
In particular, properties of graphs which are monotone and additive can be tested by using $O(d/\eps)$ queries to a partition oracle, a tool introduced by Hassidim et al.~\cite{HKNO09}. 
Thus, an implication of the recent $\poly(d/\eps)$-query partition oracle of Kumar-Seshadhri-Stolman~\cite{KSS21} is that  monotone and additive properties are testable with $\poly(d/\eps)$-queries. 
Thus the question of designing testers with $\poly(1/\eps)$-queries remains open for properties which are not monotone or not additive, as the property of being Hamiltonian.

An additional motivation for studying Hamiltonicity in minor-free graphs is, as shown by Yoshida-Ito~\cite{YI10} and Goldreich~\cite{Gol20a}, that it can not be tested with sublinear query complexity in general bounded degree graphs.

\subsection{Our Results}
\iffalse
We begin by defining a relaxed version of the partition oracle introduced by~\cite{}.
Given query access to a minor-free graph $G=(V,E)$ with degree bounded by $d$, and a parameter $\eps$, the partition oracle provides access to a partition of $V$, $\calP$, namely on query $v\in V$, the oracle returns the part of $v$ in $\calP$.
The requirements are that the partition is determined by $G$ and the randomness of the oracle and should be such that the size of each part is at most polynomial in $1/\eps$ and $d$ and the number of cut-edges of the partition is at most $\eps |V|$.  

We define a covering partition oracle in a similar manner only that on query $v\in V$ the oracle is required to return a subset $S$ of size polynomial in $1/\eps$ and $d$ such that $P \subseteq S$ where $P$ is the part of $v$ in $\calP$.
\fi

All our algorithms work under the promise that the input graph is minor-free.

\subsubsection{Testing Hamiltonicity.}
In the general graph model, we provide an algorithm for approximating the distance from Hamiltonicity up to an additive error of $\eps n$ where $n$ denotes the number of vertices in the input graph with query complexity which is $\poly(1/\eps)$ and time complexity which is exponential in $\poly(1/\eps)$.
This also implies a tolerant tester with two-sided error for testing Hamiltonicity with the same complexities.

In the bounded-degree graph model, we provide an algorithm for testing Hamiltonicity with one-sided error with query complexity which is 
$\poly(d/\eps)$, where $d$ denotes the bound on the degree, and time complexity which is exponential in $\poly(1/\eps)$.

\subsubsection{Local algorithm for constructing  spanning subgraphs of almost optimum weight.}
In the general graph model, we provide a local algorithm for constructing a sparse spanning subgraph of weight at most $(1+\eps){\rm OPT}$ where ${\rm OPT}$ denotes the weight of the MST of the input graph. 
The algorithm receives as parameters $\eps$ and an upper bound, $W$, on the maximum weight of an edge in the graph.
Moreover, the number of edges of the output graph that do not belong to the MST of the input graph\footnote{Without loss of generality we assume that the weights of the edges are distinct and hence that there is a unique MST.} is $O(\eps n/W)$.
The query complexity and time complexity of the algorithm is $\poly(W/\eps)$.
We note that in addition to incidence queries our algorithm also use random neighbour queries. 

In the bounded-degree graph model, we provide a simpler algorithm with the same guarantees whose query complexity and time complexity is $\poly(dW/\eps)$, where the polynomial of the complexity is somewhat improved compared to the algorithm for graphs of unbounded degree. 

\subsubsection{Testing monotone and additive properties of graphs}
\sloppy
We prove that any property which is monotone (closed under removal of edges and vertices) and additive (closed under the disjoint union of graphs) can be tested in the bounded degree model with $\poly(d/\eps)$-query complexity under the promise that the input graph is minor-free. The same result was recently shown independently by Kumar-Seshadhri-Stolman~\cite{KSS21}. While in~\cite{KSS21} they use partition oracles in the proof, we use a relaxed notion of partition oracles (which we introduce in this paper) and consequently obtain a somewhat improved polynomial dependence in the query complexity and a simpler algorithm.

\def\relatedwork{
%%%%%%%%%%%%%% fix
\subsection{Partition Oracles}
Partition oracles were introduced by Hassidim et al.~\cite{HKNO09} as a tool for approximating parameters and testing properties of minor-free bounded degree graphs. 
The query complexity of the partition oracle of~\cite{HKNO09} is exponential in $1/\eps$. 
The query complexity was later improved in~\cite{LR15} to be quasi-polynomial in $1/\eps$. 
Very recently, Kumar-Seshadhri-Stolman~\cite{KSS21} obtained a partition oracle with query complexity which is polynomial in $1/\eps$.

Edelman et al.~\cite{EHNO11} obtained a partition oracle with query complexity polynomial which is in $1/\eps$ for graphs with bounded treewidth.

\subsection{Testing Hamiltonicity}
Yoshida and Ito~\cite{YI10} and more recently Goldreich~\cite{Gol20a} proved a linear (in the number of vertices) lower bound for testing Hamiltonicity (even with two-sided error) in bounded degree graphs. 
Adler and Köhler~\cite{AK21} provided a deterministic construction of families of graphs for which testing 
Hamiltonicity with one-sided error requires linear number of queries.

\subsection{Testing properties of minor-free graphs}
Newman and Sohler~\cite{NS13} showed that any property of hyperfinite graphs and in particular minor-free graphs can be tested with query complexity that depends only on $1/\eps$ and $d$ where $d$ is a bound on the maximum degree. In fact, they proved a stronger claim, that a minor-free graph can be learned up to a precision of $\eps d n$ edges with such query complexity.
However, although the query complexity of their canonical tester is independent of $n$ it is super-polynomial in $d/\eps$.

For minor-free graphs of unbounded degrees, Czumaj et al.~\cite{CMOS19} obtained an algorithm whose query complexity depends only on $1/\eps$ for testing Bipartiteness. More recently, this result was generalized by  Czumaj ans Sohler~\cite{CS19} who proved that any property of minor-free graphs can be tested with one-sided error with query complexity that depends only on $1/\eps$ if and only if it can be reduced to testing for a finite family of finite forbidden subgraphs. Czumaj et al.~\cite{CMOS19} also provide a canonical tester for testing $H$-subgraph freeness for any fixed $H$ with query complexity that is independent of $n$, however super polynomial in $1/\eps$. 

It was shown that for other restrictive families of graphs of unbounded degree that every property is testable with query complexity which is at most polylogarithmic in $n$~\cite{KY14, Ito16, BKN16}.
Specifically, Kusumoto and Yoshida~\cite{KY14} proved that any property of forests can be tested with query complexity $\poly(\log n)$ and that testing Isomorphism of forests requires $\Omega(\sqrt{\log n})$.
This result was generalized in Babu-Khoury-Newman~\cite{BKN16} for $k$-outerplanar graphs.

\subsection{Local algorithms for constructing sparse spanning subgraphs}
The model of {\em local computation algorithms} as considered in
this work, was defined by Rubinfeld et al.~\cite{RTVX} (see also Alon et al.~\cite{ARVX12} and survey in~\cite{LM17}).
The problem of constructing sparse spanning subgraphs in this model was studied in several papers~\cite{LRR14,LMRRS,LR15,LL17,PRVY19,LRR20}.  
This problem is a special case of constructing an $\eps$-almost MST in which the weights of all the edges are identical.

For restricted families of  graphs, it was shown that the complexity of the problem is independent of $n$.
Specifically, it was shown in~\cite{LMRRS} that for families of graph that are, roughly speaking, sufficiency non-expanding, one can provide an algorithm with query complexity that is independent of $n$ (however, super-exponential in $1/\eps$).
This is achieved by simulating a localized version of Kruskal's algorithm.
On the negative side, it was also shown in~\cite{LMRRS} that for graphs with expansion properties that are a little better, there is {\em no local algorithm\/} that inspects a number of %vertices and
edges that is independent of $n$.

In~\cite{LRR20} there is an algorithm for locally constructing sparse spanning subgraphs in minor-free, unbounded degree, graphs with query complexity and time complexity which are polynomial in $d$ and $1/\eps$. Thus our algorithm for unbounded degree graphs generalizes this result for the weighted case.  
 
In~\cite{LRR14, LRR14-full} it was shown that a spanning subgraph of almost optimum weight can be constructed locally in minor-free graph with degree bounded by $d$ with query complexity and time complexity which are quasi-polynomial in $d$, $1/\eps$ and $W$ where $W$ is the maximum weight of an edge. 
Thus our algorithm for unbounded degree graphs generalizes this result to unbounded degree graphs and improves the complexity of the upper bound from quasi-polynomial to polynomial in $d$, $1/\eps$ and $W$.

}
\ifnum\random=0
\subsection{Related Work}

\relatedwork
\fi

\subsection{Our algorithms for minor-free unbounded degree graphs}

\subsubsection{Testing Hamiltonicity}
We begin by proving that the distance from Hamiltonicity of any graph $G=(V, E)$ equals the size of the minimum path cover of $G$ minus $1$, where a path cover of a graph is a set of disjoint paths such that each $v\in V$ belongs to exactly one path (see Claim~\ref{clm:pc}).

We then prove that if we remove $O(\eps |V|)$ edges from $G$ as well as edges that are incident to $O(\eps |V|)$ vertices in $G$ then the distance from Hamiltonicity may be increased by (at most) $O(\eps |V|)$ (see Claims~\ref{clm:edges} and~\ref{clm:vertices}).

Thus, in order to obtain an approximation, with an additive error of $O(\eps |V|)$, to the size of the minimum path cover of $G$ (and hence to its distance from Hamiltonicity) it suffices to obtain such approximation to the size of the minimum path cover of $\hat{G}$ where $\hat{G}$ is defined as follows. 
We obtain $\hat{G}$ from $G$ by first removing the edges incident to vertices of high degree, which we refer to as {\em heavy} vertices, then running the partition oracle on the resulting graph and then removing the cut-edges of the partition. 
An approximation to the size of the minimum path cover of $\hat{G}$ can be obtained by sampling vertices u.a.r. from $V$ and computing the size of the minimum path cover of their connected component in $\hat{G}$.

Since we obtain an approximation algorithm for the distance from being Hamiltonian we also obtain a tolerant tester with two-sided error for this property.

\subsubsection{Constructing spanning subgraphs with almost optimum weight}
We present our algorithm as a global algorithm and prove its correctness.
Thereafter, we describe the local implementation of this global algorithm.

Our global algorithm proceeds as follows.
In the first step, the algorithm adds all the edges between {\em heavy} vertices to the edges of the constructed spanning subgraph, $E'$, where a heavy vertex is defined to be a vertex of degree greater than some threshold.

It then runs the partition oracle on the graph induced on the vertices that are not heavy, i.e., {\em light} vertices and adds all the cut-edges of the partition to $E'$.

In the second step, each part of the partition is partitioned into subparts by running a controlled variant of Borůvka's algorithm for finding an MST on each part independently. 
The edges spanning the sub-parts are then added to $E'$.
Then, for each each sub-part, the algorithm adds a single edge to a single heavy vertex which is adjacent to the sub-part (assuming there is one).

We prove that all the edges added in the second step belong to the minimum spanning forest (MSF) of a graph which is $O(\eps/W_G)$-close to $G$, where $W_G$ denotes the maximum weight of an edge in $G$.
Additionally we prove that if we remove $O(\eps|V|/W_G)$ edges from $G$, then the weight of the minimum spanning forest (MSF) may increase by (at most) $O(\eps|V|)$. 

The second step partitions the vertices of the graph into {\em clusters} and {\em isolated parts},
where isolated parts are parts that are not adjacent to any heavy vertex, and the clusters are defined as follows.
Each cluster contains a single heavy vertex, which we refer to a the {\em center} of the cluster and sub-parts that are adjacent in the constructed graph to the center (each sub-part is adjacent to at most a single center).

In the third step the algorithm adds edges to $E'$ between pairs of cluster that are adjacent to each other.
For every edge $\{u, v\}$ which is adjacent to two different clusters, $A$ and $B$, the algorithm runs another algorithm that samples edges incident to $A$ and $B$ and returns the lightest one. The edge $\{u, v\}$ is added to $E'$ if it is lighter than the returned edge. 

We note that the algorithm that samples edges incident to a pair of specific clusters, $A$ and $B$ may not sample sufficient number of edges or may not return any edge (this is likely when the degree of both centers is large compared to the number of edges which are incident to both $A$ and $B$). 
In the analysis which is adapted from~\cite{LRR20}, we show that nonetheless, the number of edges added in the third step is sufficiently small with high probability.  The main idea is to consider the graph in which each cluster is contracted into a single vertex and then to analyse the sampling algorithm with respect to this graph which is also minor-free and hence has bounded arboricity~\footnote{The arboricity of a graph is the minimum number of forests into which its edges can be partitioned.}. The bounded arboricity of the contracted graph ensures that with high probability the sampling algorithm samples enough edges which are incident to $A$ and $B$ as long as the cut between these clusters is sufficiently large. On the other hand, if this is not the case then we show that we can afford to add to $E'$ all the edges in the cut. 

The local implementation of the above-mention global algorithm is quite straightforward and is presented in Section~\ref{sec:localimp}.

\subsection{Our algorithms for minor-free bounded degree graphs}
\subsubsection{Covering partition oracles}

We introduce a relaxed version of partition oracles which we call {\em covering partition oracles} and design such an oracle for minor-free graphs with query complexity $\poly(d/\eps)$. 
Given query access to a graph $G=(V,E)$ and parameter $\eps$ a partition oracle provides access to a partition of $V$, $\calP$, such that the size of each part of $\calP$ is small (usually polynomial in  $1/\eps$), the number of cut-edges of $\calP$ is at most $\eps|V|$ (with high probability) and $\calP$ is determined only by $G$ and the randomness of the oracle. On query $v\in V$ the oracle returns the part of $v$ in $\calP$.

A covering partition oracle has the same guarantees only that the requirement to return the part of $v$ on query $v\in V$ is relaxed as follows.
On query $v \in V$ the oracle is required to return a (small) subset $S$ such that $S$ contains the part of $v$ in $\calP$.

Our covering partition oracle builds on a central theorem from the recent work of~\cite{KSS19}. The theorem states that for any minor-free bounded degree graph there exists a partition of the graph into small parts with small number of cut-edges such that for each part of the partition, $P$, there exists a vertex, $s\in V$, such that if we preform sufficiently many (polynomial in $1/\eps$) lazy random walks from $s$ then with high probability we encounter all the vertices in $P$.
Building on this theorem we prove that the simple algorithm that on query $v\in V$ performs a set of lazy random walks from $v$ (of different lengths) and then performs a set of lazy random walks from each endpoint of a walk of the first set is a covering partition oracle.

The algorithms we in Subsections~\ref{sec:ham}-~\ref{sec:prop}, use our covering partition oracle. 

We note that since a partition oracle is a special case of a covering partition oracle (with stronger guarantees) all our algorithms for bounded degree graphs also work when one replaces calls to the covering partition oracle by calls to a partition oracle. 

As mentioned above, the use of covering partition oracles has two benefits. The first benefit is that the implementation of the covering partition oracles is much simpler and the second benefit is that the query complexity per oracle query is better (though both our covering partition oracle and the partition oracle of~\cite{KSS21} have query complexity which is $\poly(d/\eps)$), which consequently affects the query complexity of the algorithms. 

Another conceptual benefit in introducing covering partition oracles is that for some families of graphs the gap in the query complexity can be more dramatic. To give a concrete example consider $(\eps, \rho(\eps))$-hyperfinite graphs~\footnote{Let $\rho$ be a function from $\mathbb{R}_+$ to $\mathbb{R}_+$. A graph $G=(V,E)$ is $(\eps, \rho(\eps))$-hyperfinite if for every $\eps > 0$ it is possible to remove $\eps |V|$ edges of the graph such that the remaining graph has connected components of size at most $\rho(\eps)$.}.
It is straightforward to obtain a covering partition oracle with query complexity $O(d^{\rho(\eps)})$ for this family of graphs while the best known partition oracle for this family has query complexity which is $O(2^{d^{\rho(c\eps^3)}})$~\cite{HKNO09}, where $c$ is some constant.
We note that all our algorithms for bounded degree graphs work for any family of graphs for which there is a covering partition oracle (including $(\eps, \rho(\eps))$-hyperfinite graphs).

\subsubsection{Testing Hamiltonicity}\label{sec:ham}
In addition to relating the distance from Hamiltonicity of any graph $G= (V,E)$ to the size of its minimum path cover, as mentioned above, we also prove that given a subset $S \subset V$, if the size of the minimum path cover of $G[S]$ is greater than the number of edges in the cut of $S$ and $V\setminus S$ then $G$ is not Hamiltonian. Using this claim it becomes straightforward to design a one-sided error tester for Hamiltonicity that uses $O(1/\eps)$ queries to a partition oracle. We prove that it suffices to use the same number of queries to the covering partition oracle. We note that in this case there is a trade-off between the query complexity and time complexity. In particular while using the covering partition oracle rather than the partition oracle results in a better polynomial dependence of the query complexity it also results in a worse polynomial dependence in the exponent of the running time (in both cases the running time is exponential in $\poly(1/\eps)$ since we find the size of the minimum path cover by brute force~\footnote{Finding the minimum path cover is APX-hard since (as noted by Chandra Chekuri at stackexchange.com) we can reduce the TSP-path problem in metrics with distances $1$ and $2$ to it. The latter problem is APX-hard~\cite{EK01}).}).

\subsubsection{Constructing spanning subgraphs with almost optimum weight}\label{sec:MST}
As mentioned-above, given a weighted graph $G=(V, E, w)$ if we remove $O(\eps |V|/W_G)$ edges from $G$ then the weight of the MSF of the resulting graph may increase by at most $O(\eps |V|)$ compared to the weight of $G$.  
Thus given access to a partition of $V$ such that the subgraph induced on each part is connected and the number of cut-edges is $O(\eps |V|/W_G)$ we can proceed as follows. For each part of the partition we add to $E'$ the edges of the MST of the subgraph induced on this part. In addition, we add to $E'$ the cut-edges of the partition. Consequently, the total weight of the edges in $E'$ is greater than the weight of the MST of $G$ by at most $O(\eps |V|/W_G)$.
Hence, if on query $\{u, v\}$ we query the partition oracle on $u$ and $v$ then it is possible to determine whether $\{u, v\} \in E'$ where $E'$ is constructed as described above with respect to the partition of the oracle.
We prove that the same approach works when we preform the same queries to the covering partition oracle. 

\subsubsection{Testing monotone and additive properties}\label{sec:prop}
One of the main applications of the partition oracle is a general reduction for testing monotone and additive properties of bounded degree minor-free graphs. 
The idea of the reduction (from testing to the partition oracle)
is to sample $O(d/\eps)$ vertices and for each vertex $v$ in the sample to test whether the subgraph induced on the part of $v$ has the properties. The tester accepts if and only if all sampled parts pass the test. 
We prove that the same reduction works when we replace the queries to the partition oracle to queries by the covering partition oracle.

\ifnum\random=1
\subsection{Organization}
In Section~\ref{sec:unbounded} we present our algorithms for unbounded degree graphs. Due to space limitations, all our algorithms for bounded degree graphs, including our covering partition oracle, appear in the appendix.
An extended section on related work as well as omitted proofs and details of Section~\ref{sec:unbounded} appear in the appendix as well.  
\fi

\section{Preliminaries}\label{sec:prel}
In this section we introduce several definitions and some known results that will be used
in  the following sections.
Unless stated explicitly otherwise, we consider simple graphs, that is, with no
self-loops and no parallel edges. 

Let $G = (V,E)$ be a graph over $n$ vertices.
Each vertex $v\in V$ has an id, $id(v)$, where there is a full order over the ids.

The total order over the vertices induces a total order (ranking) $\rho$ over
the edges of the graph in the following straightforward manner:
$\rho((u,v)) < \rho((u',v))$ if and only if $\min\{u,v\} < \min\{u',v'\}$
or $\min\{u,v\} = \min\{u',v'\}$
and $\max\{u,v\} < \max\{u',v'\}$ (recall that $V = [n]$).
Thus, given a weighted graph, we may assume without loss of the generality that the weights of the edges are unique by breaking ties according to the order over the edges.

For a subset of vertices $X$, we let $G[X]$ denote the subgraph of $G$ induced by $X$.

When we consider bounded degree graphs, we consider the bounded-degree graph model~\cite{GR02}.
The graphs we consider have a known degree bound $d$,
and we assume we have query access to their incidence-lists representation. 
Namely, for any vertex $v$ and index $1 \leq i \leq d$ it is possible to obtain
the $i^{\rm th}$ neighbor of $v$ (where if $v$ has less than $i$ neighbors, then
a special symbol is returned). If the graph is edge-weighted, then the weight of
the edge is returned as well.

When we consider graphs with unbounded degree we consider the general graph model~\cite{PR02} equipped with an additional type of query: random neighbor query. Namely, when we query any given vertex $v$ a random neighbor of $v$ is returned~\footnote{We note that we do not use the random neighbor query in our tester for Hamiltonicity.}.

For a graph $G = (V,E)$ and two sets of vertices $V_1,V_2 \subseteq V$, we let
$E^G(V_1,V_2)$ denote the set of edges in $G$ with one endpoint in $V_1$ and one endpoint in $V_2$.
That is $E(V_1,V_2) \eqdef \{(v_1,v_2)\in E:\; v_1 \in V_1,v_2\in V_2\}$.
If $G$ is clear from the context we may omit it from the notation.

%%%%%%%%%%%%%%%
\iffalse
\BD  
\label{def:part-contract}
Let $G=(V, E, w)$ be 
an edge-weighted graph and let 
$\calP = (V_1, \ldots, V_t)$ be a partition of the vertices of $G$ such that 
for every $1 \leq i \leq t$, the subgraph induced by $V_i$ is connected.
Define the {\em contraction\/} $G/ \calP$ of $G$ with respect to the partition 
$\calP$ to be the edge-weighted graph $G' = (V', E', w')$ where:
\BE
\item  $V' = \{V_1, \ldots, V_t\}$ (that is, there is a vertex in $V'$ for
each subset of the partition $\calP$);
\item $(V_i, V_j) \in E'$ if and only if $i \neq j$ and  $E(V_i,V_j) \neq \emptyset$;
% $(v_i, v_j) \in E$ for some $v_i \in V_i$, $v_j\in V_j$,
\item $w'((V_i, V_j)) = \sum_{(u,v) \in E(V_i,V_j)} w((u,v))$.
\EE
\ED
As a special case of Definition~\ref{def:part-contract} we get the standard
notion of a single-edge contraction.
\BD\label{def:edge-contract}
 Let $G=(V, E, w)$ be an edge-weighted graph on $n$ vertices $v_1, \ldots, v_n$, and
 let $(v_i, v_j)$ be an edge of $G$. The graph obtained from $G$ by {\em contracting}
 the edge $(v_i,v_j)$ is $G/\mathcal{P}$ where $\mathcal{P}$ is the partition of 
 $V$ into $\{v_i, v_j\}$ and singletons $\{v_k\}$ for every $k\neq i, j$. 
 \ED
\fi

\subsection{Partition oracles and covering partition oracles}
\BD
For $\eps \in (0,1]$, $k \geq 1$ and a graph $G = (V,E)$,
we say that a partition $\calP = (V_1,\dots,V_t)$ of $V$ is an {\em $(\eps,k)$-partition}
(with respect to $G$), if the following conditions hold:
\BE
\item For every $1 \leq i \leq t$ it holds that $|V_i|\leq k$; 
\item For every $1 \leq i \leq t$ the subgraph induced by $V_i$
in $G$ is connected;
\item The total number of edges whose endpoints are in different parts of the partition
is at most $\eps|V|$ 
(that is, $\left|\left\{(v_i,v_j)\in E:\;v_i\in V_j,v_j\in V_j,i\neq j\right\}\right| \leq \eps|V|$).
\EE
\ED

\iffalse
\BD
For $\eps_1, \eps_2 \in (0,1]$, $\Delta\geq 1$ , $k \geq 1$ and a graph $G = (V,E)$,
we say that a partition $\calP = (\{s\}_{i\in S}, V_1,\dots,V_t)$ of $V$ is an {\em $(\eps_1, \eps_2, \Delta, k)$-partition}
(with respect to $G$), if the following conditions hold:
\BE
\item For every $1 \leq i \leq t$ it holds that $|V_i|\leq k$; 
\item For every $1 \leq i \leq t$ the subgraph induced by $V_i$
in $G$ is connected;
\item The total number of edges whose endpoints are in different parts $V_i$ and $V_j$
is at most $\eps|V|$ 
(that is, $\left|\left\{(v_i,v_j)\in E:\;v_i\in V_j,v_j\in V_j,i\neq j\right\}\right| \leq \eps|V|$);
\item For all $v \in V\setminus S$, the degree of $v$ is at most $\Delta$.
\item $|S| \leq \eps_2 |V|$.  
\EE
\ED
\fi

Let $G=(V,E)$ be a graph and let $\calP$ be a partition of $V$. 
We denote by $g_{\calP}$ the function from $v\in V$ to $2^{V}$ (the set of all subsets of $V$), 
that on input $v\in V$, returns the subset $V_{\ell} \in \calP$ such that $v\in V_{\ell}$.
We denote the set of cut-edges of $\calP$ by $E^G_\calP = \{(u,v)\in E : g_{\calP}(v) \neq g_{\calP}(u)\}$ (we may omit $G$ from the notation when it is clear from the context). 
\BD[\cite{HKNO09}]\label{def:po}
An oracle $\mathcal{O}$ is a {\em partition oracle\/} if,
given query access to the incidence-lists representation of a graph
$G=(V,E)$, the oracle $\mathcal{O}$ provides query access to
a partition $\mathcal{P} = (V_1, \ldots, V_t)$ of $V$, 
where $\mathcal{P}$ is determined by $G$ and the internal randomness of the oracle. 
Namely, on input $v\in V$, the oracle returns $g_{\calP}(v)$
and for any sequence of queries, $\mathcal{O}$ answers consistently with the 
same $\calP$. 
An oracle $\mathcal{O}$ is an $(\eps,k)$-{\em partition oracle with respect
to a class of graphs $\mathcal{C}$}
if the partition $\mathcal{P}$ it answers according to has the following properties.
\BE
\item For every $V_{\ell} \in \calP$ , $|V_{\ell}| \leq k$ and the subgraph induced 
by $V_{\ell}$  in $G$ is connected.
\item If $G$ belongs to $\mathcal{C}$, then 
$|E_\calP| \leq \eps |V|$ with high constant probability,
where the probability is taken over the internal coin flips of $\mathcal{O}$.
\EE
\ED

\noindent
We consider the following relaxation of Definition~\ref{def:po}. 
\BD
An oracle $\mathcal{O}$ is a {\em covering partition oracle\/} if,
given query access to the incidence-lists representation of a graph
$G=(V,E)$, the oracle $\mathcal{O}$, on input $v\in V$, returns a subset $S\subseteq V$ such that $g_{\calP}(v) \subseteq S$ where $\mathcal{P}$ is a partition of $V$ determined by $G$ and the internal randomness of the oracle. 
For any sequence of queries, $\mathcal{O}$ answers consistently according to the 
same
$\calP$. 
An oracle $\mathcal{O}$ is an $(\eps, k)$-{\em covering partition oracle with respect
to a class of graphs $\mathcal{C}$}
if the following conditions hold.
\BE
\item On every query, the subgraph induced by the subset returned by  $\mathcal{O}$, $S$, is connected and $|S| \leq k$.
\item If $G$ belongs to $\mathcal{C}$, then with high probability,
$|E_\calP| \leq \eps |V|$,
where the probability is taken over the internal coin flips of $\mathcal{O}$.
\EE
\ED

%By the above definition, if $G \in \mathcal{C}$, then with high constant
%probability the partition $\calP$ is an $(\eps,k)$-partition, while if $G \notin \mathcal{C}$
%then it is only required that each part of the partition is connected and has size at most $k$.
%We are interested in partition oracles that have small query complexity, 
%namely, that perform few queries to the graph (for each vertex they are queried on).

\subsection{Graph minors}

Recall that a graph $R$ is called a {\em minor} of a graph $G$ if $R$ is isomorphic to a graph that can be obtained by zero or more edge contractions on a subgraph of $G$. A graph $G$ is {\em $R$-minor-free} if $R$ is not a minor of $G$.
We next quote two results that will play a central role in this work.
\BF\label{fct:forest}
Let $R$ be a fixed graph with $r$ edges.
For  every $R$-minor-free graph $G=(V,E)$ it holds that:
\begin{enumerate}
\item $|E| \leq r \cdot |V|$; 
\item $E$ can be partitioned into at most $r$ forests.  
\end{enumerate}
\EF

Unless stated otherwise, in all our algorithms, we assume that the input graph is $R$-minor-free graph where $R$ is a fixed graph with $r$ edges (we could receive $r$ as a parameter but we make this assumption for the sake of brevity). 

\subsection{Hamiltonian path and minimum path cover} 
\begin{definition}[Hamiltonian path]
A Hamiltonian path in $G=(V,E)$ is a path between two vertices of $G$ that visits each vertex of $G$ exactly once.
\end{definition}

\begin{definition}[minimum path cover]\label{def:mpc}
Given an undirected graph $G=(V,E)$, a path cover is a set of disjoint paths such that every vertex $v \in V$ belongs to exactly one path. The minimum path cover of $G$ is a path cover of $G$ having the least number of paths.
\end{definition}

\subsection{Local algorithms for constructing sparse spanning subgraphs}\label{sub:preMST}
\BD[\cite{LRR20}]
\label{dfn:SSG-alg}
An algorithm $\calA$ is a local sparse spanning graph (LSSG) algorithm if, given
$n\geq 1$, $\eps > 0$,
and query access to the incidence-lists representation of a connected graph $G=(V,E)$ over $n$ vertices,
it provides
oracle
 access to a subgraph $G'=(V, E')$ of $G$ such that:
\BE
\item\label{it:connect} $G'$ is connected. 
\item\label{it:internal-rand} $\s{E'} \leq (1+\eps)\cdot n$ with high constant probability~\footnote{In some papers the required success probability is high, i.e. at least
    $1-1/\Omega(n)$.},
where $E'$ is determined by $G$ and the internal randomness    of $\mathcal{A}$.
\EE
More specifically, on query ${u,v}\in E$,
    $\calA$ returns whether $(u,v) \in E'$, and for any
    sequence of edges, $\calA$ answers consistently with
    the same $G'$.

An algorithm $\calA$ is an {\em LSSG
algorithm for a family of graphs $\mathcal{C}$}
if the above conditions hold, provided that the input graph
$G$ belongs to $\mathcal{C}$.
\ED

\BD[\cite{LRR14}]\label{def:mwsg}
A local algorithm for $(1+\eps)$-approximating the minimum weight spanning graph of a graph $G=(V, E, w)$ with positive weights and $\min_{e\in E} w(e) \geq 1$, is a local algorithm for $(1+\eps)$-sparse spanning graph of $G=(V, E, w)$ for which the following holds:
$\sum_{e\in E'} w(e) \leq (1+\eps)\alpha$, where $\alpha$ is the weight of a minimum weight spanning tree of $G$.
\ED
For a graph $G=(V, E, w)$ we define $W_G = \max_{e\in E} w(e)$ (when it is clear from the context, we sometimes omit the subscript $G$). 
We denote by $\MSF(G)$ the set of edges the minimum-spanning-forest of $G$ (as mentioned above we assume without loss of generality that all weights are distinct and thus the minimum-spanning-forest is unique).
For a connect weighted graph we denote by $\MST(G)$ the set of edges the minimum-spanning-forest of $G$.
For a subset of edges $S\subseteq E$, we define $\w(S) \eqdef \sum_{e\in S} w(e)$.

\medskip
\noindent
Our algorithms build on the following rules.  
\begin{enumerate}
    \item The {\em cut rule} states that for any cut of the graph (a cut is a partition of the vertices into two sets), the lightest edge that crosses the cut must be in the MST.

\item The {\em cycle rule} states that if we have a cycle, the heaviest edge on that cycle cannot be in the MST.

\end{enumerate}

\section{Algorithms for minor-free graphs with unbounded degrees}\label{sec:unbounded}

\subsection{Testing Hamiltonicity}
 
In this section we prove the following Theorem.

\begin{theorem}\label{thm:hamU}
Given query access to an input graph $G=(V,E)$ where $G$ is a minor-free unbounded degree graph and parameters $\eps$ and $|V|$, there exists an algorithm that accepts $G$ with probability at least $2/3$ if $G$ is $\eps/2$-close to being Hamiltonian and rejects $G$ with probability at least $2/3$ if $G$ is $\eps$-far from being Hamiltonian.  
The query complexity of the algorithm is $\poly(1/\eps)$ and the running time is exponential in $\poly(1/\eps)$.
\end{theorem}

Theorem~\ref{thm:hamU} is a direct consequence of following claim (which is proved in the sequel).

\begin{claim}\label{clm:app}
Given an input graph $G=(V, E)$ which is a minor-free graph, and parameter $\eps$, Algorithm~\ref{alg:ham2} outputs a value $x$ such that with high constant probability $\delta_{\ham}(G) -\eps|V| \leq  x \leq \delta_{\ham}(G) +\eps|V|$. 
\end{claim}

\ifnum\random=0
We next prove a couple of claims which we use in the proof of Claim~\ref{clm:app}. 
\else
We next state a couple of claims which we use in the proof of Claim~\ref{clm:app}.
\fi
\begin{claim}\label{clm:pc}
Let $G=(V,E)$ be a graph and let $k$ be the size of a minimum path cover of $G$. Then the distance of $G$ for being Hamiltonian, $\dist(G)$, is $k-1$.
\end{claim}

\def\proofclaim{
Let $G=(V,E)$ be a graph and let $\mathcal{C} = \{P_1,\ldots, P_k\}$ be a minimum path cover of $G$.

We first prove that $\dist(G) \leq k-1$.
For every $1\le i \le k-1 $ we add an edge which connects the end-vertex of $P_i$ to the start-vertex of $P_{i+1}$. Thus, by adding $k-1$ edges we constructed a Hamiltonian path in $G$.

We next prove that $\dist(G) \geq k-1$.
By definition, there exist $\dist(G)$ edges such that when added to $G$, $G$ becomes Hamiltonian.
Let $E'$ denote a set of $\dist(G)$ such edges and let $G' = (V, E\cup E')$ be the graph resulting from adding these edges to $G$.
Let $\mathcal{H} = (s_1, \ldots, s_{|V|})$ denote a Hamiltonian path in $G'$. 
After we remove back the edges in $E'$ we break $\mathcal{H}$ into $|E'|+1$ connected components (each edge we remove adds an additional connected component), i.e. into $|E'|+1$ paths. Thus the size of the minimum path cover of $G$ is at most $|E'|+1 = \dist(G) +1$. Thus $k \leq \dist(G) +1$ and so $\dist(G) \geq k-1$ as desired.
}

\ifnum\random=0
\begin{proof} 
\proofclaim
\end{proof}
\fi

\ifnum\random=0
The next claims state that if we remove $\eps$-fraction of the edges or edges that are incident to $\eps$-fraction of the vertices then the distance from being Hamiltonian is increased by at most $O(\eps |V|)$.
\fi

\begin{claim}\label{clm:edges}
Let $G= (V, E)$ be a graph and let $F \subseteq E$ be a subset of edges.  
Then $$\delta_{\ham}(G) \leq \delta_{\ham}(G') \leq \delta_{\ham}(G) + |F|,$$ where $G' = (V, E')$ and $E' = E\setminus F$. 
\end{claim}

\def\proofclaimTwo{
The claim that $\delta_{\ham}(G) \leq \delta_{\ham}(G')$ follows from the fact that the distance from being Hamiltonian can not decrease when we remove edges.  

Let $\calC$ be a minimum path cover of $G$.
By Claim~\ref{clm:pc}, $\delta_{\ham}(G) = |\calC|-1$.
Now consider removing the edges in $F$ one by one and how this affects the number of paths in $\calC$.
After removal of a single edge, the number of paths may increase by at most one. 
Thus, after removing all the edges in $F$ the paths in $\calC$ break into at most $|\calC| + |F|$ paths. 
Thus the size of the minimum path cover of $G'$ is at most $|\calC| + |F|$. By claim~\ref{clm:pc}, $\delta_{\ham}(G') \leq |\calC| + |F| - 1 = \delta_{\ham}(G) + |F|$, as desired.
}

\ifnum\random=0
\begin{proof} 
\proofclaimTwo
\end{proof}
\fi

\begin{claim}\label{clm:vertices}
Let $G= (V, E)$ be a graph and let $S \subseteq V$ be a subset of vertices.  
Then $$\delta_{\ham}(G) \leq \delta_{\ham}(G') \leq \delta_{\ham}(G) + 2|S|,$$ where $G' = (V, E')$ and $E'$ is the set of edges in $E$ that are not incident to vertices in $S$. 
\end{claim}

\def\proofclaimThree{
The proof of this claim is similar to the proof of Claim~\ref{clm:edges}.

The claim that $\delta_{\ham}(G) \leq \delta_{\ham}(G')$ follows from the fact that the distance from being Hamiltonian can not decrease when we remove edges.  

Let $\calC$ be a minimum path cover of $G$.
By Claim~\ref{clm:pc}, $\delta_{\ham}(G) = |\calC|-1$.
Now consider removing the edges adjacent to vertices in $S$ vertex by vertex and how this affects the number of paths in $\calC$.
After removal of edges incident to a specific vertex, the number of paths may increase by at most two.
This follows from the fact that each vertex $v$ belongs to exactly one path, $P$, and the fact that when the edges incident to $v$ are removed, $P$ may break into at most $3$ different paths.   
Thus, after removing all the edges incident to vertices in $S$ the paths in $\calP$ break into at most $|\calC| + 2|S|$ paths. 
Thus the size of the minimum path cover of $G'$ is at most $|\calC| + 2|S|$. By claim~\ref{clm:pc}, $\delta_{\ham}(G') \leq |\calC| + 2|S| - 1 = \delta_{\ham}(G) + 2|S|$, as desired.
}

\ifnum\random=0
\begin{proof} 
\proofclaimThree
\end{proof}
\fi

\begin{algorithm}
\caption{Approximating the distance to Hamiltonicity in minor-free, unbounded degree, graphs\label{alg:ham2}}
\textbf{Input:} Oracle access to a minor-free, unbounded-degree, graph $G=(V,E)$\\
\textbf{Output:} $(1+\eps)$-approximation to $\delta_{\ham}(G)$
\BE
\item Define $\Delta \eqdef 8 c(h)/\eps$, $H$ to be the set of vertices of degree greater than $\Delta$, and $L \eqdef V\setminus H$.  
\item Sample a set $S$ of $y = \Theta(1/\eps^2)$ vertices u.a.r.
\item For each vertex $v \in S$:\label{step:ham2:1}
\BE
\item If $v \in L$ then:  
\BE
\item Query the partition oracle on $v$ with parameter $\eps/4$ with respect to the graph $G[L]$. Let $S_v$ denote the returned set.\label{step.po}
\item Set $x_v = k/|S_v|$ where $k$ is the size of the minimum path cover of $G[S_v]$.
\EE 
\item Otherwise, set $x_v = 1$.
\EE
\item Output $\frac{\sum_{v\in S} x_v}{|S|} \cdot |V|$. \label{step3ham}
\EE

\end{algorithm}

\begin{proof}[Proof of Claim~\ref{clm:app}]

Let $\calP$ denote the partition for which the partition oracle executed in Step~\ref{step.po} answers according to. 
With high constant probability $|E_{\calP}| \leq \frac{\eps|V|}{4}$.
Let $E_1$ denote this event. 

Let $G'=(V, E')$ be the graph such that $E'$ is the set of edges that are not incident to vertices in $H$ and are not in $E_{\calP}$.
By Claims~\ref{clm:edges} and~\ref{clm:vertices}, 
$$
\delta_{\ham}(G) \leq \delta_{\ham}(G') \leq \delta_{\ham}(G) + |E_{\calP}| + 2|H|. 
$$
By Markov's inequality and Fact~\ref{fct:forest}, $|H| \leq \frac{\eps |V|}{4}$.
We prove that, conditioned that $E_1$ occurs,  Algorithm~\ref{alg:ham2} outputs with high constant probability a $(1+\eps)$-approximation to $\delta_{\ham}(G')$. 

For each $v\in V$ define the random variable $x_v$ as defined in Step~\ref{step:ham2:1} of Algorithm~\ref{alg:ham2}.
Let $T \in \calP$ be a part in $\calP$, then $\sum_{v\in T} x_v = k$ where $k$ is the minimum path cover of $G[T]$.
Thus $\sum_{v\in V} x_v$ is the minimum path cover of $G'$.
Since for every $v\in V$, $x_v \in (0, 1]$, it follows by the additive Chernoff's bound that with high constant probability $\left|\frac{\sum_{v \in S} x_v}{|S|} - \frac{\sum_{v\in V} x_v|}{|V|}\right| \leq \frac{\eps}{4}$.
Thus, with high constant probability,
$$\delta_{\ham}(G) - \frac{\eps |V|}{4} \leq \frac{\sum_{v \in S} x_v}{|S|} \cdot |V| \leq \delta_{\ham}(G) + \eps|V|,$$
as desired.
\end{proof}

\subsection{A Local algorithm for constructing a spanning subgraph with almost optimum weight}\label{sec:MSTU}

In this section we prove the following theorem.

\begin{theorem}\label{thm:MSTU}
There exists a local algorithm for $(1+\eps)$-approximating the minimum weight spanning graph for the family of unbounded degree minor-free graphs, with positive weights and minimum weight which is at least $1$. The query complexity and time complexity of the algorithm is $\poly(W/\eps)$ where $W$ is an upper bound on the maximum weight.
The algorithm receives $\eps$ and $W$ as parameters.
\end{theorem}

\ifnum\random=0
Before we describe our algorithm and prove its correctness we prove the following useful claim which states that if we remove $\eps$-fraction of the edges then the weight of the MSF of the resulting graph does not increase by much (compared to the original graph).
\else
Our algorithm builds on the following claim.
\fi

\begin{claim}\label{claim:almost}
Let $G = (V, E, w)$ be a weighted graph and let $G' = (V, E', w)$ be a graph such that $E' = E \setminus S$ where $S\subseteq E$.
Then $\w(\MSF(G')) \leq \w(\MSF(G)) + |S| W_G$.  
Moreover, $|\MSF(G')\setminus \MSF(G)| \leq |S|$.
\end{claim}

\def\claimMSF{
We claim that $\MSF(G) \subseteq \MSF(G') \cup S$.
To see this observe that for every edge $e \in E' \setminus \MSF(G')$, it holds, by the cycle rule, that there exists a cycle in $G'$ such that $e$ is the heaviest edges in this cycle. Thus, these edges are not in $\MSF(G)$ either (because all the cycles in $G'$ exist in $G$ as well).
%Thus $\MSF(G) \subset \MSF(G') \cup S$ as claimed.

Since the number of connected components in $G$ is at most the number of connected components in $G'$ it holds that $|\MSF(G')| \leq |MSF(G)|$.
Thus, 
\begin{equation}
   |\MSF(G) \setminus \MSF(G')| \geq |\MSF(G') \setminus \MSF(G)|\label{eq.1}\;.
\end{equation}

Since $\MSF(G) \subseteq \MSF(G') \cup S$ it holds that $\MSF(G) \setminus MSF(G') \subseteq S$. 
Thus,
\begin{equation}
    |\MSF(G) \setminus \MSF(G')| \leq |S|\label{eq.2}\;.
\end{equation}
It follows from Equations~\ref{eq.1} and~\ref{eq.2} that $|\MSF(G')\setminus \MSF(G)| \leq |S|$. 
Thus, the claim follows from the bound on the maximum weight of an edge in $G$.
}
\ifnum\random=0
\begin{proof}
\claimMSF
\end{proof}
\fi

\noindent
We next describe our algorithm from a global point of view.

\subsubsection{The global algorithm}

Our global algorithm, which is listed in Algorithm~\ref{alg:MstGlob}, proceeds as follows.
In Step~\ref{global:step1}, the algorithm adds all the edges between {\em heavy} vertices to $E'$ where a heavy vertex is defined to be a vertex of degree greater than $\Delta \eqdef 6r^2 W/\eps$.

It then runs the partition oracle on the graph induced on the {\em light} vertices, namely vertices that are not heavy, and adds all the cut-edges of the partition, $\calP$, to $E'$ (Step~\ref{global:step2}).

Step~\ref{global:step3} consists of two parts. 
In Sub-step~\ref{step:MST51} the algorithm runs Algorithm~\ref{alg:sub-parts} on each part in $\calP$.
Algorithm~\ref{alg:sub-parts} partitions the parts into connected sub-parts by preforming a controlled variant of Borůvka's algorithm. The edges added in Sub-step~\ref{step:MST51} are the edges that span the sub-parts. 
In Sub-step~\ref{step:MST53}, for each sub-part, the algorithm adds a single edge to a single vertex in $H$ which is adjacent to the sub-part (assuming there is one).

This partitions the vertices of the graph into {\em clusters} and {\em isolated parts}, namely, parts in $\calP$ that are not adjacent to any vertex in $H$.
Each cluster contains a single heavy vertex, which we refer to as the {\em center} of the cluster and sub-parts that are connected by a single edge to the center.

In Step~\ref{global:step4} the algorithm adds edges to $E'$ between pairs of clusters that are adjacent to each other in $G' = (V, E\setminus E_{\calP})$. 
For every edge $\{u, v\}$ which is adjacent to two different clusters, $A$ and $B$, the algorithm runs Algorithm~\ref{alg:sample} which samples edges incident to $A$ and $B$ and returns the lightest one. The edge $\{u, v\}$ is added to $E'$ if it is lighter than the edge returned by Algorithm~\ref{alg:sample} (or if the algorithm returned null). 

We note that we do not have direct access to uniform samples of edges incident to a pair of specific clusters, $A$ and $B$. In fact, it could be the case that Algorithm~\ref{alg:sample} does not return any edge (this is likely when the degree of both centers is large compared to the number of edges which are incident to both $A$ and $B$). By adapting the analysis in~\cite{LRR20}, we show that nonetheless, the number of edges added in Step~\ref{global:step4} is sufficiently small with high probability.     

This concludes the description of the global algorithm. We next prove its correctness.
\ifnum\random=0
The local implementation of the algorithm appears in the appendix (Algorithm~\ref{alg:MstU}).
\fi
\begin{algorithm}
\caption{Global algorithm for approximated-MST in unbounded-degree minor-free graphs \label{alg:MstGlob}}
\textbf{Input:} parameters $\eps$ and $W$ and access to a minor-free graph $G=(V,E)$.\\
\textbf{Output:} $G=(V, E')$ which is an approximated-MST of $G$. 
\begin{enumerate}
    \item Let $H$ denote the set of all vertices of degree greater than $\Delta$ and let $L = V \setminus H$. 
    \item Add all the edges of $G[H]$ to $E'$.\label{global:step1}
    \item Run the partition oracle with parameter $\eps/(6W)$ on all vertices in $G[L]$. Let $\calP$ denote the resulting partition.
    \item Add all the edges in $E_\calP$ to $E'$.\label{global:step2}
    \item For each $S \in \calP$: \label{global:step3}
    \begin{enumerate}
        \item Run Algorithm~\ref{alg:sub-parts} and let $F=(S, A)$ denote the returned graph.\label{step:MST51}
        \item Add the edges in $A$ to $E'$. \label{step:MST52}
        \item For each connected component of $F$, $C$:\label{step:MST53}
        \begin{enumerate}
            \item Add to $E'$ the lightest edges which is adjacent to $C$ and $H$ (if such edge exists). 
        \end{enumerate}
    \end{enumerate}
    \item For each $v \in H$, define the {\em cluster} of $v$, denoted by $\calC(v)$, to be subset of vertices that contains $v$ and all the vertices in sub-parts $B$ such that there exists an edge in $E'$ that is incident to $v$ and a vertex in $B$. $v$ is referred to as the {\em center} of the cluster.
    \item For each edges $\{u, v\}\in E$ such that $u$ and $v$ belong to different clusters:\label{global:step4}
    \begin{enumerate}
        \item Run Algorithm~\ref{alg:sample} and add $\{u, v\}$ to $E'$ if it is lighter than the edge returned by the algorithm or if it returned null.
    \end{enumerate}
    \item Return $G'=(V, E')$.
\end{enumerate}
\end{algorithm}

\begin{claim}\label{claim:edges1}
With high constant probability, the number of edges added to $E'$ in steps~\ref{global:step1} and ~\ref{global:step2} of Algorithm~\ref{alg:MstGlob} is at most $\eps |V|/(3W)$. 
\end{claim}

\def\claimSim{
With high constant probability $|E_\calP| \leq \eps |V|/(6W)$. Since $G$ is minor-free it follows by Fact~\ref{fct:forest} and Markov's inequality the number of edges in $G[H]$ is at most $\eps|V|/(6W)$. The claim follows. 
}
\ifnum\random=0
\begin{proof}
\claimSim
\end{proof}
\fi

\begin{claim} \label{claim:edges3}
All edges added to $E'$ in step~\ref{global:step3} of Algorithm~\ref{alg:MstGlob} belong to $\MSF(\hat{G})$ where $\hat{G}=(V, E\setminus E_{\calP})$.
\end{claim}

\def\claimSimB{
Let $S\in \calP$ and let $F=(S, A)$ denote the graph returned by Algorithm~\ref{alg:sub-parts}.

We first prove that all the edges in $A$ belong to $\MSF(\hat{G})$.
Since each sub-part $B$ of $S$ is active (see Step~\ref{step:active} of Algorithm~\ref{alg:sub-parts}) as long as the lightest edge in $E^G(B, S\setminus B)$ is lighter than the lightest edge in $E^G(B, H)$ it follows that $e_B$ (see Step~\ref{step:subedge} of Algorithm~\ref{alg:sub-parts}) is the lightest edge in $E^{\hat{G}}(B, V\setminus B)$. Thus, all the edges of $A$ belong to $\MSF(\hat{G})$ by the cut rule (see Subsection~\ref{sub:preMST}).  

We next prove that for each connected component of $F$, $C$, the lightest edge which is adjacent to $C$ and $H$ (if such edge exists) is in $\MSF(\hat{G})$.
We first note that we only need to consider $S$ such that $E^G(S, H) \neq \emptyset$. 
In this case each connected component of $F$, $C$, is adjacent to at least one vertex in $H$. Since $C$ is not active it follows that lightest edge which is adjacent to $C$ and $H$ is the lightest edge in the cut of $C$ in $\hat{G}$. Thus the claim follows from the cut rule. 
This concludes the proof of the claim.
}

\ifnum\random=0
\begin{proof}
\claimSimB
\end{proof}

In the proof of the following claim we closely follow the analysis in~\cite{LRR20} (see proof in the appendix).
\fi
\begin{claim}
\label{clm:F}
Let $U$ denote the set of edges in $E'$ that are incident to two different clusters (namely, each endpoint belongs to a different cluster). With probability $1-1/\Omega(|V|)$, $|U| \leq \eps|V|/(3W)$\;.
\end{claim}

\begin{claim}
Let $G = (V, E)$ be a connected minor-free graph, then with high constant probability the graph returned by Algorithm~\ref{alg:MstGlob}, $G' = (V, E')$, is connected and $\sum_{e\in E'} w(e) \leq (1+\eps){\rm OPT}$, where ${\rm OPT}$ is the weight of a minimum weight spanning tree of $G$.
\end{claim}

\begin{proof}
We begin by proving $G'$ is connected. To see this observe that each vertex either belongs to a cluster or to a subset $S \in \calP$ such that $E^G(S, H) = \emptyset$.

By construction, the subgraph induced on each cluster in $G'$ is connected.
For any two clusters which are connected by an edge the lightest edge that connects the clusters belongs to $E'$.
This follows from Step~\ref{global:step4} in Algorithm~\ref{alg:MstGlob}.

If $E^G(S,H) = \emptyset$ then $G'[S]$ is connected by Algorithm~\ref{alg:MstGlob}, as all sub-parts of $S$ remain active throughout the entire execution of the algorithm.
Moreover, all the edges in $E^G(S, V\setminus S)$ are in $E'$ as well.
Thus $G'$ is connected.

The claim regrading the weight of the edges of $G'$ follows from Claims~\ref{claim:almost}-~\ref{clm:F}.
\end{proof}

\begin{algorithm}

\caption{Return sampled lightest edge \label{alg:sample}}
\textbf{Input:} $a\in H$ and  $b\in H$.\\
\textbf{Output:} The sampled lightest edge between $\calC(a)$ and $\calC(b)$ 
\begin{enumerate}
    \item Initially $A = \emptyset$ and let $x$ be an upper bound on the size of the parts returned by the partition oracle when executed with parameter $\eps/(6W)$.
    \item Sample a set, $S$, of $\tilde{\Theta}(W^2r^3x \Delta/\eps^2)$ edges incident to $a$.\label{sample:step1}
    \item For each $\{a, u\} \in S$ do: \label{sample:step2}
    \BE
    \item If $u \in H$ then add it to $A$ if and only $u = b$.
    \item Otherwise, find the part of $u$ and the sub-parts of this part.
    \item For each sub-part find its center.
    \item Add to $A$ all the edges between the sub-part of $u$ and sub-parts that belong to $\calC(b)$.
    \EE
    \item Repeat Steps~\ref{sample:step1}-\ref{sample:step2} where $a$ and $b$ switch roles. 
    \item Return the lightest edge in $A$ (if $A = \emptyset$ then return null).
\end{enumerate}

\end{algorithm}

\begin{algorithm}[H]
\caption{Partition into sub-parts \label{alg:sub-parts}}
\textbf{Input:} Access to the input graph $G=(V,E)$ and a subset $S \subseteq L$ such that $G[S]$ is connected.\\
\textbf{Output:} A graph $F = (S, A)$ such that each connected component of $F$ is a sub-part of $S$ 
\BE
\item Initially every vertex in $S$ is a in its own sub-part (a singleton) and $A = \emptyset$.
\item We say a sub-part $B$ is {\em active} if the lightest in $E^G(B, S\setminus B)$ is lighter than the lightest edge in $E^G(B, H)$ or if $E^G(B, H) = \emptyset$.\label{step:active}
\item While there are still active sub-parts do:
\BE
\item Each active sub-part $B$ selects the lightest edge in $E^G(B, S\setminus B)$, denoted by $e_B$.\label{step:subedge}
\item For each active sub-part, $B$, add $e_B$ to $A$ 
\item Update the new sub-parts to be the connected components of the graph $F = (S, A)$ (each connected component is a sub-part).
\EE
\item Return $F$.
\EE
\end{algorithm}

\section{Algorithms for minor-free graphs with bounded degrees}

\subsection{Covering partition oracle}
In this section we prove the following theorem.
\begin{theorem}\label{thm:cpo}
Algorithm~\ref{alg:cpo} is an 
$(\eps, \poly(\eps^{-1}))$-covering-partition oracle for minor-free bounded degree graphs with query complexity $\poly(\eps^{-1})$. Specifically, the size of the sets returned by the oracle is $O(\eps^{-640}\log^2 (1/\eps))$.
\end{theorem}

We begin with a couple of definitions and lemmas from~\cite{KSS21} that we build on.

\begin{definition}[\cite{KSS21}] \label{def:clip} Given $\bx \in (\mathbb{R}^+)^{|V|}$ and parameter $\xi \in [0,1)$, the $\xi$-clipped
vector $\clip{\bx}{\xi}$ is the lexicographically least vector $\by$ optimizing the program:
$\min \|\by\|_2$, subject to $\|\bx - \by\|_1 \leq \xi$ and $\forall v \in V, \by(v) \leq \bx(v)$.
\end{definition}

\begin{lemma}[\cite{KSS21}]\label{lem:clip:new}
There is an absolute constant $\alpha$ such that the following holds.
Let $H$ be a graph on $r$ vertices. Suppose $G$ is a $H$-minor-free
graph. Then for any $h \geq \alpha r^3$, there exists 
at least $(1 - 1/h )n$ vertices such that
$\| \clip{\wvec{v}{\len}}{3/8}\|_2^2 \geq h^{-7}$.
	\end{lemma}

Given two parameters $\eps \in [0, 1/2]$, and a graph $R$ on $r \geq 3$ vertices. The length of the random walk is $\ell = \alpha r^3 +\lceil\eps^{-20}\rceil $ where $\alpha$ is some absolute constant.
\begin{theorem}[\cite{KSS21}]\label{thm:partition}
Suppose there are at least $(1-1/\len^{1/5})n$ vertices $s$ such that
$\|\clip{\wvec{s}{\ell}}{1/4}\|^2_2 > \len^{-c}$. Then, there 
is a partition $\{P_1, P_2, \ldots, P_b\}$ of the vertices such that:
\begin{enumerate}
    \item For each $P_i$, there exists $s \in V$ such that: $\forall v \in P_i$, $\sum_{t < 10\len^{c+1}} \prw{s}{v}{t} \geq 1/8\len^{c+1}$.
    \item The total number of edges crossing the partition is at most $8dn\sqrt{c\len^{-1/5}\log \len}$. 
\end{enumerate}
\end{theorem}

\BC\label{col:1}
Let $G = (V, E)$ be a graph which is $R$-minor-free where $R$ is a graph on $r$ vertices.
There exists a partition $\{P_1, P_2, \ldots ,P_b\}$ of the $V$ such that:
\begin{enumerate}
    \item For each $P_i$, there exists $s \in V$ such that: $\forall v \in P_i$,
$\sum_{t < 10\ell^{8}} p_{s,t}(v) \geq 1/8\ell^{8}$.
\item The total number of edges crossing the partition is at most $\eps d n$.
\end{enumerate}

\EC

\begin{proof}
We first note that $8dn\sqrt{c\len^{-1/5}\log \len} \leq \eps d n$ for sufficiently large constant $\alpha$.

By Lemma~\ref{lem:clip:new} there exist 
at least $(1 - 1/\ell )n$ vertices such that
$\| \clip{\wvec{v}{\len}}{3/8}\|_2^2 \geq \ell^{-7}$.
Thus the corollary follows from the facts that $(1 - 1/\ell )n \geq (1-1/\ell^{1/5})n$ and $\|\clip{\wvec{s}{\ell}}{1/4}\|^2_2 \geq \|\clip{\wvec{s}{\ell}}{3/8}\|^2_2$. 
\end{proof}

\begin{algorithm}[h]
\caption{Covering-partition oracle\label{alg:cpo}}
\textbf{Input:} $v\in V$.\\
\textbf{Output:} A subset $S$ which covers the part of $v$. 
\BE
\item For every $t < 10\ell^{8}$ perform $x \eqdef \Theta(\ell^{8}\log \ell)$ random walks of length $t$ from $v$.\label{step:1}\;
\item Let $R$ denote the endpoints of the random walks preformed in the previous step.
\item For every vertex $r \in R$, for every $t < 10\ell^{8}$, perform $x$ random walks of length $t$ from $r$. \label{step:2}
\item Let $S$ denote the set of all vertices encountered by the random walks performed in Step~\ref{step:1} and Step~\ref{step:2}.
\item Return $S$.
\EE

\end{algorithm}

\begin{proof}[Proof of Theorem~\ref{thm:cpo}]
Let $G = (V, E)$ be a graph which is $R$-minor-free where $R$ is a graph over $r$ vertices.
Consider the partition of $V$, $\mathcal{P} = \{P_1, P_2, \ldots ,P_b\}$ as defined in Corollary~\ref{col:1} when we take the proximity parameter to be $\eps/2$.
We shall define another partition $\mathcal{P'}$ which is a refinement of $\mathcal{P}$ such that Algorithm~\ref{alg:cpo} returns for every $v \in V$, a subset $S$ such that $P' \subseteq S$ where $P'$ denotes the part of $v$ in $\mathcal{P}'$. Thereafter, we shall prove that, with high probability, the number of cut-edges of $\mathcal{P}'$ is not much greater than the number of cut-edges of $\mathcal{P}$.  

For every $v \in V$, we say that $v$ {\em fails} if Algorithm~\ref{alg:cpo}, when queried on $v$, does not return $S$ such that $P^v \subseteq S$, where $P^v$ denotes the part of $v$ in $\mathcal{P}$.
We define $\mathcal{P'}$ as follows.
For every $v \in V$, if there exists $u\in P^v$ such that $u$ fails, then the part of $v$ in $\mathcal{P'}$ is defined to be the singleton $\{v\}$ (namely, the entire part $P^v$ is partitioned into singletons in $\mathcal{P}'$).
Otherwise, it is defined to be $P^v$.

We next show that with high probability the cut-edges of $\mathcal{P}'$ is at most $\eps d |V|$.
For every $v$, the probability that $v$ fails is at most $p \eqdef \ell^{-c_1}$ for an appropriate setting of $x$ (with accordance to the Theta-notation), where $c_1$ is a constant that will be determined later.
Let $y = \ell^{c_2}$ be an upper bound on the number of vertices in the parts of $\mathcal{P}$, where $c_2$ is a constant.~\footnote{Clearly, since for each $P_i$, there exists $s \in V$ such that: $\forall v \in P_i$,
$\sum_{t < 10\ell^{8}} p_{s,t}(v) \geq 1/8\ell^{8}$, it follows that $y \leq 10\ell^{8} \cdot 8\ell^{8}$.} 

For every $v\in V$, define the random variable $X_v$ as follows. If $v$ fails then $X_v = |P_v|/y$ and otherwise $X_v = 0$.
Clearly $y \cdot \sum_{v \in V} X_v \geq |\mathcal{P'}| - |\mathcal{P}|$.
Note that $\{X_v\}_{v\in V}$ are independent random variables ranging in $[0, 1]$.

For every $v\in V$, we define the random variable $Y_v$ as follows. With probability $p$, $Y_v = 1$ and otherwise $X_v = 0$.
Clearly, $Y_v$ dominates $X_v$.
Since $\{Y_v\}_{v\in V}$ are identical independent random variables, it follows by the multiplicative Chernoff's bound that with high probability $y \cdot \sum_{v \in V} Y_v \leq y\cdot |V|\cdot 2p$.

Since for sufficiently large $c_1$, $p \leq \frac{\eps}{4y}$, it follows that with high probability $|\mathcal{P'}| - |\mathcal{P}| \leq \eps |V|/2$. Thus, the number of cut-edges in $\mathcal{P'}$ is greater than the number of cut-edges in $\mathcal{P}$ by at most $\eps d |V|/2$ (recall that the refinement from $\mathcal{P}$ is done by decomposing entire parts into singletons) .

\end{proof}

\subsection{Testing Hamiltonicity}
 In this section we prove the following theorem.
 
 \begin{theorem}\label{thm:ham1}
Given query access to an input graph $G=(V,E)$ where $G$ is a minor-free bounded degree graph and a parameters $\eps$ and $|V|$, Algorithm~\ref{alg:ham1} accepts $G$ with probability $1$ if $G$ is Hamiltonian and rejects $G$ with probability at least $2/3$ if $G$ is $\eps$-far from being Hamiltonian. The query complexity of the algorithm is $\poly(d/\eps)$ and the running time is exponential in $\poly(d/\eps)$. 
\end{theorem}

The correctness of Algorithm~\ref{alg:ham1} builds on the following claim which, given $S\subset V$, bounds the size of a minimum path cover in $G[S]$ by the size of the cut of $S$.  

\begin{claim}\label{hamcut.cl}
Let $G=(V,E)$ be a graph and let $S\subset V$ be a subset of vertices of $G$. Let $k$ be the size of a minimum path cover in $G[S]$. 
If $k-1 > |E(S, V\setminus S)|/2$, then there is no Hamiltonian path in $G$.
Moreover, any Hamiltonian path in $G$ must include at least $2(k-1)$ edges from $E(S, V\setminus S)$.
\end{claim}

\begin{proof}
Let $G=(V,E)$ be a graph.
Assume toward contradiction that there exists Hamiltonian path in $G$, $\mathcal{H}=(v_1, v_2, \ldots, v_{|V|})$ and a subset $S\subset V$ such that  $k-1 > |E(S, V\setminus S)|/2$, where $k$ is the size of a minimum path cover in $G[S]$. Let $\mathcal{P'}_S$ denote the set of all maximal sub-paths of $\mathcal{H}$ in $G[S]$. 
In order to connect the sub-paths in $\mathcal{P'}_S$ it must hold that $\mathcal{H}$ leaves and returns to $G[S]$ at least $2(\mathcal{P'}_S - 1)$ times, each time using a different edge. Thus, $|E(S, V\setminus S)| \geq 2(|\mathcal{P'}_S| - 1)$. Since $\mathcal{P'}_S$ is a path cover of $G[S]$ it follows that $|\mathcal{P'}_S| \geq k$, thus, $|E(S, V\setminus S)| \geq 2(k - 1)$, in contradiction to our assumption.
Hence the claim follows. 
\end{proof}

We next list our algorithm and prove its correctness.

\begin{algorithm}
\caption{Testing Hamiltonicity in minor-free, bounded degree, graphs\label{alg:ham1}}
\textbf{Input:} Oracle access to a minor-free, bounded-degree, graph $G=(V,E)$\\
\textbf{Output:} Tests if $G$ is Hamiltonian with one-sided error. 
\BE
\item Sample a subset, $S\subseteq V$, of $y \eqdef \Theta(x/\eps)$ vertices, uniformly at random, where $x$ is an upper bound on the size of the sets returned by the covering partition oracle when execute with parameter $\eps/6$.
\item For each $v\in S$ do:
\BE
\item Query the covering partition oracle on $v$ with parameter $\eps/6$, and let $S_v$ denote the returned set.\label{step.po}
\item If $E(S_v, V\setminus S_v) = \emptyset$ then return REJECT.\label{step1.ham1}
\item For each subset $T \subseteq S_v$ such that $G[T]$ is connected, find the size of the minimum path cover of $T$ and return REJECT if it is greater than $|E(T, V\setminus T)|/2 + 1$. \label{step2.ham1}
\EE
\EE
\end{algorithm}

\begin{proof}[Proof of Theorem~\ref{thm:ham1}]
By Claim~\ref{hamcut.cl}, Algorithm~\ref{alg:ham1} never rejects graphs which are Hamiltonian.

Let $G$ be a minor-free bounded degree graph which is $\eps$-far from being Hamiltonian. We shall prove that Algorithm~\ref{alg:ham1}  rejects $G$ with probability at least $2/3$.
Let $\calP$ denote the partition that the oracle, executed in Step~\ref{step.po}, answers according to.
With high constant probability, it holds that $|E_\calP| \leq \frac{\eps |V|}{6} $.
Let $E_1$ denote the event that this conditions holds.

Let $\calF$ denote the set of parts, $S$, in $\calP$, such that $E(S, V\setminus S) = \emptyset$ or for which the size of the minimum path cover of $G[S]$ is greater $|E(S, V\setminus S)|/2 + 1$

Assume towards contradiction that $E_1$ occurs and that $|\calF| < \eps |V| / (2 x)$ where $x$ is an upper bound on the number of vertices in each part of $\calP$.
We next show that $\dist(G) \leq \eps |V|$ in contradiction to our assumption. 

For each part $S \in \calF$ we construct a path over $S$ that visits each vertex in $S$ exactly once by adding at most $|S|-1 \leq x -1$ edges to $G$. 

For each part $S \in \calP\setminus \calF$, let $\mathcal{C}_S$ denote a minimum path cover of $G[S]$.
We construct a path over $S$ that visits each vertex in $S$ exactly once by adding at most $|\mathcal{C}_S| - 1 \leq |E(S, V\setminus S)|/2$ edges to $G$. 

We then connect all the paths induced on the different parts of $\calP$ by adding at most $|\calP| = |\calF| + |\calP\setminus \calF|$ edges.

Overall the number of edges added is at most:
\begin{equation*}
    |\calF| \cdot (x-1) + \left|E_\calP\right| + |\calF| + |\calP\setminus \calF| 
    \leq |\calF| \cdot x + 3\left|E_\calP\right| < \eps|V| ,
\end{equation*}
where the first inequality follows from the fact that $|\calP\setminus \calF| \leq 2|E_\calP|$ as for each $S \in \calP\setminus \calF$ it holds that $E(S, V\setminus S) \cap E_\calP \neq \emptyset$ and each edge in $E_\calP$ is adjacent to at most $2$ parts in $\calP$. 

Thus, if $\delta_{\ham(G)} > \eps |V|$ and $E_1$ occurs then $|\calF|  \geq \eps|V|/(2 x)$. Thus, the number of vertices in parts that belong to $\calF$ is at least $\eps|V|/(2 x)$, which implies that $G$ is rejected with high probability either in Step~\ref{step1.ham1} or in Step~\ref{step2.ham1} of Algorithm~\ref{alg:ham1}, as desired.  
\end{proof}

\subsection{Local algorithms for constructing a spanning subgraph with almost optimum weight}

In this section we prove the following theorem. 

\BT\label{thm:appmst}
Algorithm~\ref{alg:Mst} is a local algorithm for $(1+\eps)$-approximating the minimum weight spanning graph for minor-free graphs, with high constant success probability and time and query complexity $\poly(W,d,\eps^{-1})$.
\ET

\begin{proof}
Let $\calP$ denote the partition that the oracle, executed in Step~\ref{step:po1}, answers according to.
With high constant probability, it holds that $|E_\calP| \leq \eps |V|/W$.
Let $E_1$ denote the event that this conditions holds.
We claim that the number edges for which Algorithm~\ref{alg:Mst} returns YES for which both endpoints belong to the same part is at most $|V|-1$.
To see this consider a part $T \in \calP$ and a cycle $C$ in $G[T]$. Let $\{u, v\}$ denote the heaviest edge in the cycle. When queried on $u$ and $v$ the covering partition oracle returns sets $S_u$ and $S_v$ such that $T \subseteq S_u \cup S_v$. Thus the cycle $C$ is contained in $G[S_u \cup S_v]$. Therefore the algorithm returns NO on $\{u, v\}$ in Step~\ref{step:mst3}. By the cycle rule the number of edges in $G[T]$ for which the algorithm returns YES is exactly $|T|-1$. 

Hence conditioned on $E_1$, the total number of edges for which Algorithm~\ref{alg:Mst} returns YES is at most $(|V|-1) + \eps |V|/W$. 

By the cycle rule, any edge, $e$, for which Algorithm~\ref{alg:Mst} returns NO does not belong to the MST of $G$.
Since the MST consists of exactly $|V|-1$ edges, it follows that, conditioned on $E_1$, the number of edges that do not belong to the MST and for which Algorithm~\ref{alg:Mst} returns YES is at most $\eps |V|/W$ as desired.
\end{proof}

\begin{algorithm}
\caption{Local algorithm for approximated-MST in bounded-degree minor-free graphs \label{alg:Mst}}
\textbf{Input:} $\{u, v\}\in E$ and parameters $\eps$ and $W$.\\
\textbf{Output:} YES if $\{u, v\}$ belongs to the approximated-MST and NO otherwise. 
\BE
\item Perform a query $u$ and a query $v$ to the covering-partition oracle with parameter $\eps/W$. Let $S_u$ and $S_v$ denote the subsets returned by the oracle, respectively.\label{step:po1}
\item Find the subgraph induced on $S_u \cup S_v$, denoted by $G[S_u \cup S_v]$.
\item Return NO if and only if $\{u, v\}$ is the heaviest edge on any cycle in $G[S_u \cup S_v]$.  \label{step:mst3}
\EE

\end{algorithm}

\subsection{Testing monotone and additive properties}
\begin{theorem}
Any property of graphs which is monotone (closed under removal of edges and vertices) and additive (closed under the disjoint union of graphs) can be tested with one-sided error in minor-free graphs with bounded degree $d$ with query complexity which is $\poly(d/\eps)$ where $\eps$ is the proximity parameter. 
\end{theorem}

\begin{proof}
Let $\calT$ be a property of graphs which is monotone and additive. 
We propose the following algorithm for testing $\calP$ on an input graph $G$ which is a minor-free graphs of degree bounded by $d$. Sample a set of $O(d/\eps)$ vertices, $S$, uniformly at random and run the covering partition oracle on each $v \in S$ with parameter $\eps/2$. 

For each $v \in S$, let $S_v$ denote the set returned by the covering partition oracle when queried on $v$. Return ACCEPT iff for all $v \in S$, $G[S_v]$  has the property $\calT$.

If $G$ has the property $\calT$ then since $\calT$ is monotone it follows that for all $v\in S$, $G[S_v]$ has the property $\calT$ as well.

Let $\calP$ denote the partition that the covering partition oracle answers according to. 
With high constant probability $|E_\calP| \leq \eps |V|/2$.
Let $E_1$ denote the event that this condition holds.
If $G$ is $\eps$-far from having the property $\calT$ then, conditioned on $E_1$, $G' = (V, E\setminus E_\calP)$ is $(\eps/2)$-far from having the property $\calT$.

Thus we need to remove at least $\eps |V|/2$ edges from $G'$ to obtain the property $\calT$. 
By the additivity and monotonicity of $\calT$ it follows that $G'$ has the property $\calT$ if and only if for every $T\in \calP$, $G[T]$ has the property $\calT$.
Thus, there are at least $\eps |V|/2$ edges, and hence at least $\eps |V|/(2d)$ vertices, that belong to parts, $T \in \calP$ such that $G[T]$ does not have the property $\calT$. Hence, with high constant probability the algorithm sample one of these vertices and rejects $G$.
This concludes the proof.
\end{proof}

\subsection*{Acknowledgement}
We would like to thank Dana Ron and Oded Goldreich for helpful comments.

 \bibliographystyle{plain}
\bibliography{refs}

\ifnum\random=1
\appendix
\section{Related Work}
\relatedwork
\fi

\ifnum\random=0
\appendix
\fi
\section{Omitted proofs and details}

\subsection{Proof of Claim~\ref{clm:F}}
%\begin{proof}[Proof of Claim~\ref{clm:F}]

%For a cluster $B$, let $\partial(B)$ denote the set of edges that are incident to $B$ (namely, edges for which exactly one endpoint is in $B$)\mnote{remove heavy heavy?}.
% For each heavy vertex $u$, we charge to its cluster, $\calC(u)$, a subset of the edges in $\partial(\calC(u)) \cap F$
For each cluster $B$, we charge to $B$  a subset of the edges incident to $B$
so that the union of all the charged edges (over all clusters) contains $U$.
Our goal is to show that %w.h.p.\mnote{D: quantify}
with probability $1-1/\Omega(n)$,
the total number of charged edges is at most $\eps |V|/(3W)$.
%In what follows, let $\eps' = \eps/XX$.\mnote{D: fill XX (16?)}
%\mnote{D: say here upper bound on the size of each part}

Let $\hat{G} = (V, \hat{E})$ be such that $\hat{E} = E\setminus E_\calP$. 
Consider the  auxiliary graph, denoted $\wtG$, that results from contracting each cluster $B$ and isolated parts in $\hat{G}$
into a  {\em mega-vertex\/} in $\wtG$, which we denote by $v(B)$.
For each pair of clusters $B$ and $B'$ such that $E^{\hat{G}}(B,B')$ is non-empty, there is an edge $(v(B),v(B'))$ in  $\wtG$, which we refer to as a {\em mega edge\/}, and whose weight is $|E^{\hat{G}}(B,B')|$. Since $G$ is minor-free, so is $\wtG$.
By Fact~\ref{fct:forest}, which bounds the arboricity of minor-free graphs, we can partition the mega-edges of $\wtG$ into $r$ forests.
Consider orienting the mega-edges of $\wtG$ according to this partition (from children to parents in the trees of these forests), so that each mega-vertex has at most $r$ outgoing mega-edges.
For cluster $B$ and a cluster $B'$ such that $(v(B),v(B'))$ is an edge in $\wtG$ that is oriented from $v(B)$ to $v(B')$,
we shall charge to $B$ a subset of the edges in $E^{\hat{G}}(B, B')$, as described next.

Let $x$ be an upper bound on the size of parts returned by the partition oracle when executed with parameter $\eps/(6W)$.
Thus $x$ is an upper bound on the size of part in the partition $\calP$ of $G[\light]$.
Let $E^b(B,B')$ denote the subset of edges in $E^{\hat{G}}(B,B')$
  that are the $(\eps/(9rW ))\cdot |E(B,B')|$ lightest edges of $E(B,B')$.
We  charge all the edges in $E^b(B,B')$ to $B$. The rationale is that for these edges it is likely that the algorithm won't sample an edge in $E^{\hat{G}}(B,B')$ which is lighter.
The total number of such edges is at most $(\eps/(9 rW))\cdot |E| \leq \eps|V|/(9W)$.

%%%%%%%%%%%%%%%%%%%%%%%%%%%%%%%%%%
%%%%%%%%%%%%%%%%%%%%%%%%%%%%%%%%%%
%%%%%%%%%%%%%%%%%%%%%%%%%%%%%%%%%%

For a light vertex $y$ let $subpart(y)$ denote the subpart of $y$.
Let $u$ be the center of the cluster $B$, and let
$N^b(u,B')$ be the set of vertices, $y \in N(u)$, such that:
$$\left(y\in B' \mbox{ and } (u, y)\in E^b(B,B')\right)
     \mbox{ or } \left(\exists (y',z)\in E^b(B,B') \mbox{ s.t. } y'\in {subpart}(y)\right)\;.$$
% $$N^b(u,B') \;=\;
% \left\{y\in N(u)\,:\, \left(y\in B' \mbox{ and } (u, y)\in E^b(B,B')\right)
% \mbox{ or } \left(\exists (y',z)\in E^b(B,B') \mbox{ s.t. } y'\in {\cal P}(y)\right)\right\}\;.
% $$
That is, $N^b(u,B')$ is the subset of neighbors of $u$ such that if
%the algorithm
Algorithm~\ref{alg:sample}
selects one of them in Step~\ref{sample:step1}, then it obtains an edge in $E^b(B,B')$.
We consider two cases.

{\bf First case: $|N^b(u,B')|/|N(u)| < \eps^2/(162 W^2 r^3 x \wtd)$.}
In this case we charge all edges in $E^{\hat{G}}(B,B')$ to $B$.
%Since
For each part ${subpart}(y)$ such that $y \in N^b(u,B')$ there are at most
$|{subpart}(y)|\cdot \wtd$ edges $(y',z)\in E^b(B,B')$ for which $y'\in {subpart}(y)$.
Therefore, in this case
$|E^b(B,B')| \leq  x  \wtd \cdot |N^b(u,B')| \leq N(u) \cdot \eps^2/(162 W^2 r^3)$
and hence $|E(B,B')| < (9rW/\eps)\cdot \eps^2/(162 W^2 r^3) \cdot N(u)$.
It follows that the total number of charged edges of this type is at most
$ (\eps/(18rW))\cdot 2|E| \leq \eps n /(9W)$.

{\bf Second case: $|N^b(u,B')|/|N(u)| \geq \eps^2/(162 W^2 r^3 x \wtd)$.}
%We are left with edges in $E(B,B')\setminus E^b(B,B')$ where
For each $u$ and $B'$ that fall under this case we define the set of edges $Y(u, B') = \{(u, v): v \in N^b(u,B')\}$ and denote by $Y$ the union of all such sets (over all such pairs $u$ and $B'$).
Edges in $Y$ are charged to $B$ if and only if they belong to $U$ and are incident to a vertex in $B$.
Fix an edge in $Y$ that is incident to $B$, and note that the selection of neighbors of $u$
is done according to a $t$-wise independent distribution for $t > 4q$, where $q$ is the sample size set in
Step~\ref{sample:step1} of the Algorithm~\ref{alg:sample}.
Therefore, the probability that the edge belongs to $U$ is
upper bounded by $(1-\eps^2/(162 W^2 r^3 x \wtd))^q$, which by the setting of $q$, is at most $p=\eps/(18Wr)$ (for sufficiently large constant w.r.t. the Theta notation).

We next show, using Chebyshev's inequality, that with high probability, the number of edges in $Y$ that are in $U$ is at most $2p|E|$.
For $y \in Y$, define $J_y$  to be an indicator variable that is $1$ if and only if $y\in F$.
Then for a fixed $y\in Y$, $\E[J_y] \leq p$ and $\{J_y\}$ are pairwise independent (this is due to the fact that the samples of every pair of edges are pairwise independent).
Therefore, by Chebyshev's inequality,
$$
\Pr\left[\sum_{y\in Y} J_y \geq 2p|E|\right] \leq \frac{\Var[\sum_{y\in Y} J_y]}{(p|E|)^2} = \frac{\sum_{y\in Y} \Var(J_y)}{(p|E|)^2} \leq \frac{p(1-p)|E|}{(p|E|)^2} =  \frac{1-p}{p|E|} = \frac{1}{\Omega(n)}\;,
$$
and the proof of Claim~\ref{clm:F} is completed.

\begin{remark}
The random seed that Algorithm~\ref{alg:MstGlob} uses consists of two parts.
The first part is for running the partition oracle.
The second part is for selecting random neighbors in Step~\ref{sample:step1} of Algorithm~\ref{alg:sample}.  % and~\ref{step:sam2}.
%Since for each edge considered in Step~\ref{global:step4} of Algorithm~\ref{alg:MstGlob} we pick a random sample of
%$\tilde{\Theta}(W^2 r^3 x \Delta)$
% $\wtd \eqdef 16(c(h))^2/\eps$, $\sz = \max\{1/\gamma, k\}$, $\gamma = \eps/(8c(h))$, $c(h) = O(h\log h)$
%neighbors independently, we obtain % by Theorem~\ref{thm:twise} 
Since the selection of neighbors is according to a $t$-wise independent distribution we obtain
that a random seed of length %$\tilde{O}((h/\eps)^7\log n)$
$\tilde{O}(\log n)$
is sufficient.

\iffalse
The running time of Algorithm~\ref{alg:MstGlob} is upper bounded by a constant times the running time of Algorithm~\ref{alg:ssg} times the number of edges sampled in Step~\ref{step:sam}. % and~\ref{step:sam2},
which by Theorem~\ref{thm:main} and the setting of $q$ in the algorithm is $\poly(h/\eps)$.
This completes the proof of Theorem~\ref{thm:main1}.
\fi
\end{remark}

%%%%%%%%%%
\iffalse
% \BL
We next show that
$G'=(V,E')$ is connected and for every edge $(u, v) \in E\setminus E'$ there exists a path of length at most %$\tilde{O}(\eps/c(h))$
$\tilde{O}(h/\eps)$  %\mnote{D: check (should be consistent with thm statement}
in $G'$.
% \EL
% \BPF
By induction on the rank of the edge we obtain that indeed we add at least one edge between every pair of adjacent clusters to the spanning graph.
This implies that the graph remains connected. Moreover, the stretch factor is at most a constant factor greater than the stretch factor of Algorithm~\ref{alg:ssg}.
\fi
%%%%%%%%%

%\subsubsection{Random Seed}\label{subsubsec:random-seed}
%\mnote{D: need to connect this part to prob analysis}

\subsection{The local implementation of Algorithm~\ref{alg:MstGlob}}\label{sec:localimp}

Algorithm~\ref{alg:MstU} is the local implementation of Algorithm~\ref{alg:MstGlob} and is listed next.

\begin{algorithm}
\caption{Local algorithm for approximated-MST in unbounded-degree minor-free graphs \label{alg:MstU}}
\textbf{Input:} $\{u, v\}\in E$.\\
\textbf{Output:} YES if $\{u, v\}$ belongs to the approximated-MST and NO otherwise. 
\BE
\item If both $u$ and $v$ are in $H$ return YES.
\item If both $u$ and $v$ are light:
\BE
\item Query the partition oracle on $u$ and $v$ and return YES if they belong to different parts.
\item Find the sub-parts of $u$ and $v$ by running Algorithm~\ref{alg:sub-parts}.
\item If $u$ and $v$ are in the same sub-part:
\BE
\item Return YES if $\{u, v\}$ is in the set of edges returned by Algorithm~\ref{alg:sub-parts} (when running on $u$).
\item Otherwise, return NO.
\EE
\item Otherwise, set $C_u$ to be the center of $u$ and $C_v$ to be  the center of $v$.
\item If $C_u = C_v$ return NO.
\EE
\item Otherwise, if $u$ is light and $v$ is heavy (and analogously if $v$ is light and $u$ is heavy) then:
\begin{enumerate}
    \item Find the sub-part of $u$ and set $C_u$ to be the center of this sub-part
    \item Set $C_v = v$
\end{enumerate}
\item Run Algorithm~\ref{alg:sample} on $C_u$ and $C_v$ and return YES if the edge $\{u, v\}$ is lighter than the edge returned by the algorithm. Otherwise, return NO.
\EE
\end{algorithm}

\ifnum\random=1
\subsection{Proof of Claim~\ref{clm:pc}}
\proofclaim

\subsection{Proof of Claim~\ref{clm:edges}}
\proofclaimTwo

\subsection{Proof of Claim~\ref{clm:vertices}}
\proofclaimThree

\subsection{Proof of Claim~\ref{claim:almost}}
\claimMSF

\subsection{Proof of Claim~\ref{claim:edges1}}
\claimSim

\subsection{Proof of Claim~\ref{claim:edges3}}
\claimSimB
\fi

\end{document}